\documentclass{article} 
\usepackage{nips14submit_e,times}
\usepackage{moreverb,url}
\usepackage[pdftex]{graphicx}
\usepackage{amsmath}
\usepackage{amssymb}
\usepackage[justification=centering]{caption}
\usepackage{graphicx}

\title{Artificial Bee Colony-based Adaptive Position Control of Electrohydraulic Servo Systems with Parameter Uncertainty}

\author{
Babajide O. Ayinde \\
Electrical and Computer Engineering\\
University of Louisville\\
Louisville, KY 40218 \\
\texttt{babajide.ayinde@louisville.edu} \\
\And
Sami El-Ferik\\
Department of Systems Engineering \\
King Fahd University of Petroleum and Minerals \\
Dhahran, 31261, Saudi Arabia\\
}

\newtheorem{theorem}{Theorem}
\newtheorem{proof}{Proof}
\newtheorem{remark}{Remark}
\newtheorem{assumption}{Assumption}

\nipsfinalcopy 

\begin{document}

\maketitle

\begin{abstract}
In this paper, a robust adaptive backstepping-based controller is developed for positioning the spool valve of Electro-Hydraulic Servo System (EHSS) with parameter fluctuations. Artificial Bee Colony (ABC) algorithm is utilized to drive the parameters of the proposed controller to a good neighborhood in the solution space. The optimization problem is formulated such that both the tracking error and control signal are minimized concurrently. The results show that the proposed controller guarantees the uniform ultimate boundedness of the tracking error and control signal. Moreover, illustrative simulations validate the potentials and robustness of the proposed schemes in the presence of uncertainties. The proposed controller is also compared with sliding mode control.
\end{abstract}

\textbf{Keywords:} Electro-hydraulic systems, backstepping control, nonlinear systems, adaptive control, artificial bee colony

\section{Introduction}
The need for fast and powerful responses in many industrial applications has made the Electro-Hydraulic Servo Systems (EHSS) very popular. Although EHSS are very useful in many industrial applications, ranging from aerospace flight control to manufacturing. In seismic applications, EHSS is one of the major components in Vibroseis \cite{mintsa_feedback_2012}. However, all the aforementioned applications demand high precision control of the system's output.
\\
\indent
Many factors such as fluid inflow-outflow in the servo valve and friction on actuator-valve moving parts contribute to the nonlinearity in the system dynamics. Also, phenomena such as entrapment of air inside the hydraulic valve, parameter variations (due to temperature changes), unknown model errors and perturbations elevate the complexity of controller design. Many methodologies have been developed to address this challenging task. In spite of all the nonlinear behaviors of the EHSS, linear control theories have been used for developing reasonably working controllers. One drawback of using this approach is that linear control analysis breaks down with changing operating conditions and uncertainties. The performance of linear controller was enhanced using a feedback-feedforward iterative learning controller \cite{wang_high_2012}, and the problem of parameters variation was addressed using adaptive schemes \cite{yanada_adaptive_2007,zhang_adaptive_2009}.\\
\indent
Even though adaptivity can greatly handle the problem of parameter variation, traditional adaptive schemes suffer a serious setback especially when system to be controlled is not linear in parameters and/or if the dynamics has uncertain or unmodeled part. Unfortunately, EHSS falls under the category of such systems with parameter variations and uncertain dynamics. For instance, EHSS parameters such as supply pressure and mass of the load/piston varies with surrounding conditions. Lyapunov approach was used in \cite{mintsa_feedback_2012} to design an enhanced feedback linearization-based controller for EHSS with supply pressure uncertainties, however, the effect of unknown disturbance at level of the velocity was not taken into account. This ultimately narrows the scope of the design.\\
\indent
Backstepping is one of the design strategies that employs a progressive approach to formulate the control law \cite{rozali2013asymptotic} and it is generally used for designing stabilizing controls for some class of dynamical nonlinear systems. This method is developed by inserting new variables that depends on the state variables, controlling parameters and the stabilizing functions. The essence of this stabilizing function is to redress any nonlinearity that can impede the stability of the system. Due to the stepwise nature of the control design, formulation of the controller generally starts with a well known stable system. Subsequently, virtual controllers are employed to stabilize the outer subsystems progressively. The process of using the virtual controllers in the stabilization of the subsystem continues until the external control is accessible. In fact, it has been shown that backstepping technique can be used to force nonlinear systems to behave like a linear system transformed into a new set of coordinates.\\
\indent
One of the numerous advantages of using backstepping technique to design a controller is its ability to avoid useful nonlinearity cancelation \cite{ayinde2015backstepping,el2015backstepping}. The objective of backstepping gravitates towards stabilization and tracking in contrast with its corresponding feedback linearization method. This facilitates controller designs for perturbed nonlinear system even if the perturbation is nowhere around the equation containing the input. Backstepping method is generally used for tracking and regulation problems \cite{rozali2013asymptotic}. In \cite{dashti2014neural}, backstepping-based neural adaptive technique was used for velocity control of EHSS with internal friction, flow nonlinearity, and noise. However, the effect of external disturbance was not considered. More specifically, EHSS are prone to parameter variations due to temperature change. For instance bulk modulus and viscous friction coefficients are prone to variation due to fluctuations in temperature. Owing to this fact, the need to design a controller that adapts to these changes is of paramount importance.\\
\indent
One important consideration that has to be carefully addressed in many mechanical systems is friction \cite{hutamarn_neuro-fuzzy_2012}. A LuGre model-based adaptive control scheme was proposed to model and estimate the frictional effect \cite{wang_high_2012}. This approach does not only account for frictional effect, but also offers good disturbance rejection and robustness to uncertainties. Variable structure controllers have also been used to model both friction and load as external disturbances \cite{bonchis_variable_2001}. Feedback controller with auto-disturbance rejection has also been shown to control the position of EHSS with both internal and external disturbances \cite{xiang_electro-hydraulic_2012}. System model identification technique has also been used to estimate the model of EHSS and adaptive Fuzzy PID controller was developed for position control \cite{shao_model_2009}. The problem of external load variation and coulomb friction have been mitigated using a nonlinear adaptive feedback linearization position control scheme with load disturbance rejection and friction compensation \cite{Mintsa_2011}. Dynamic particle swarm optimization-based algorithm has been proposed to optimize the control parameters and improve the tracking performance of the closed loop system \cite{lai_synchronization_2013}.
\\
\indent
Evolutionary techniques are gaining popularity among the researchers in the last few decades due to their ability to localize global optima, or in some cases, good local optima \cite{hashim2015fuzzy,hashim2015fuzzy2}. In this work, Artificial Bee Colony (ABC) \cite{Karaboga_2005,Karaboga_2007} is proposed to tune the parameters of the controller for optimality. Significance of ABC have been shown in path planning of multi robot and other control applications \cite{bhattacharjee_multi-robot_2011}.\\
\indent
This paper proposes a backstepping-based approach to design a robust adaptive controller for a highly nonlinear EHSS - a single input single output (SISO) system. The EHSS consists of a four-way spool valve supplying a double effect linear cylinder with a double rodded piston. The piston drives a load modeled by mass, spring and a sliding viscous friction. Under the proposed adaptive controller, a practical stability of the closed loop system is ensured and the uniform ultimate boundedness of the tracking error is guaranteed. ABC is used to tune the controller parameters by minimizing both the tracking error and control signal. The rest of this paper is organized as follows: Section II presents the formulation of the problem and section III presents the adaptive controller design. Section IV discusses the ABC optimization and its implementation for the proposed adaptive backstepping controller. Section V discusses the simulation results and finally, section VI concludes the paper.
\section{System Model and Problem Formulation}
The dynamics of the EHSS are highly nonlinear due to many factors such as friction on the actuators, fluid inflow and outflow in the valve. Moreover, entrapment of air and several other phenomena such as parameter variations, unknown model errors aggravates the complexity of the control design. The EHSS model in Fig.~\ref{Fig1} with dynamics given below in \ref{MyEq1} is considered.\\
\begin{figure}
  \centering
  \includegraphics[width=8cm]{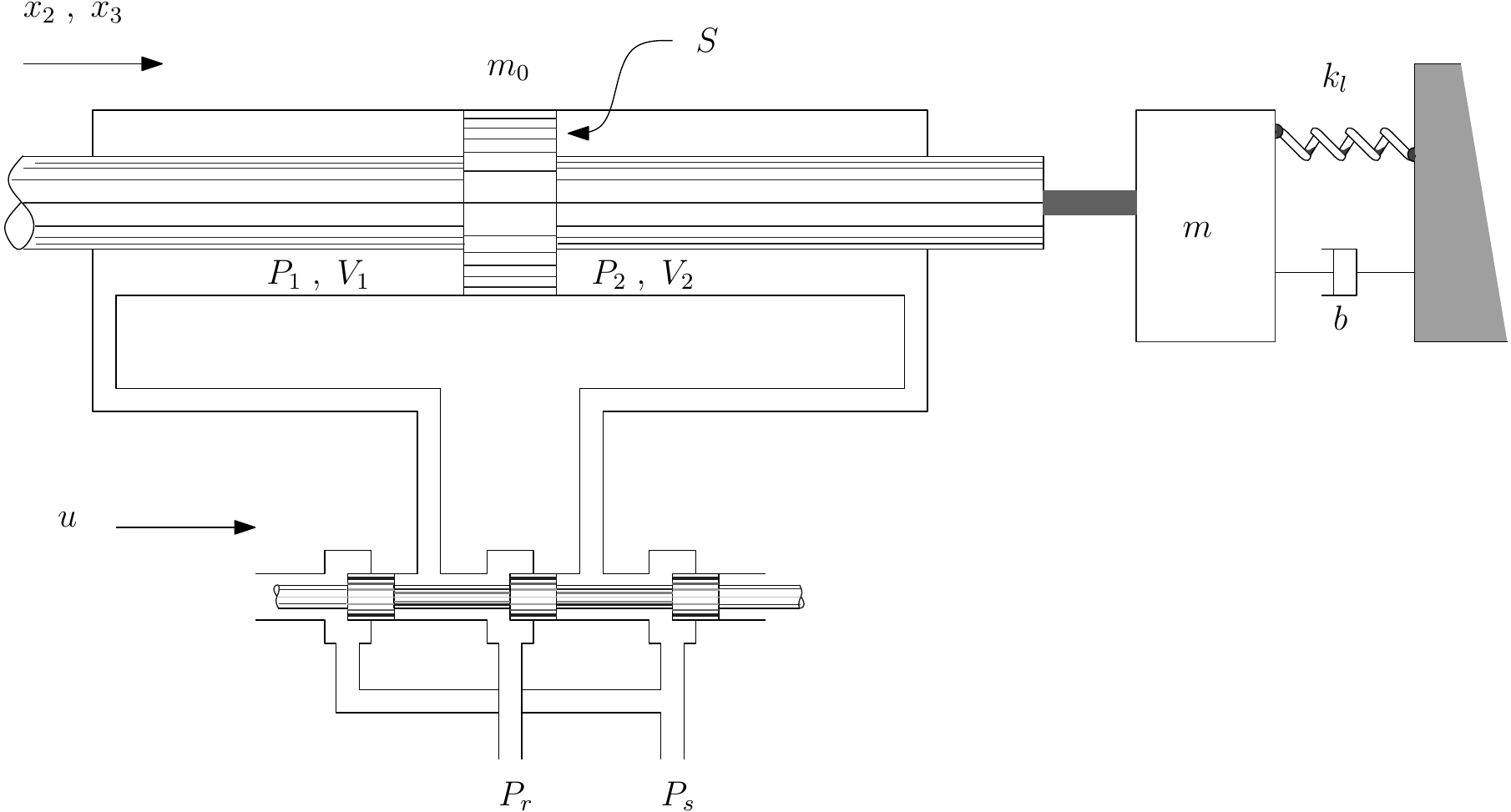}
  \caption{Electro-hydraulic System}\label{Fig1}
\end{figure}

\begin{equation} \label{MyEq1}
  \begin{split}
  {\dot x}_1 &= \tfrac{4B}{V_t}(ku\sqrt{P_d-\text{sign}(u)x_1}-\tfrac{\alpha x_1}{1+\gamma \lvert u\rvert}-Sx_2 ),\\
  {\dot x}_2 &= \tfrac{1}{m_t}(Sx_1-bx_2-\beta x_3),\\
  {\dot x}_3 &=x_2+d(t).
  \end{split}
\end{equation}

where $\beta=(k_l+\Delta\,k_l)$. $x_1$ is the differential pressure between the two chambers, $x_2$ and $x_3$ are the velocity and position of the rod respectively. $k_l+\Delta\ k_l$ denotes the uncertain spring stiffness and $b$ is the viscous damping coefficient, $V_t$ is the total volume of the forward and return chambers, $P_d$ is the supply and return pressure difference while $m_t$ is the total mass of the load and piston. $B$ and $S$ are the bulk modulus and net cross-sectional area of one side of the piston respectively. $k$, $\gamma$ , $\alpha$ are intrinsic constants of the servo valve; $\gamma$ and $\alpha$ are used to model the leakage in the servo valve.\\
 In the sequel, we consider the following assumptions hold:
\begin{assumption}
\begin{enumerate}
\item d(t) is an unknown but bounded disturbance with $|d(t)|<d_{max}$.
\item $\Delta k_l$ is unknown and bounded with $|\Delta k_l|< \Delta k_l^{max}$.
\item The dynamics of the spool-valve is assumed fast enough so it can be ignored in the dynamic model.
\item The states of the system are available for controller design.
\item Reference input ($r(t)$) is a known continuously differentiable bounded trajectory.
\end{enumerate}
\end{assumption}
The nonlinearities with respect to the input $u$ in the dynamics of the system makes it challenging to control the output of the system. We are able to circumvent this challenge by designing a backstepping based controller that will drive the position of the rod to a desired reference $r(t)$. In order to use the backstepping method to design the controller, a re-indexing of the states variables is needed to transform the system into its standard strict feedback form.
Let\\
\begin{equation} \label{MyEq2}
  \begin{split}
  \xi_1 &= x_3, \;\;
  \xi_2 \; = x_2, \;\;
  \xi_3 = x_1.
  \end{split}
\end{equation}
 The dynamics of the transformed system is then given in Eq. (\ref{MyEq4})
 \begin{equation} \label{MyEq4}
  \begin{split}
  {\dot \xi}_1 &= \xi_2+d(t),\\
  {\dot \xi}_2 &= \tfrac{1}{m_t}(S \xi_3-b\xi_2-\beta \xi_1),\\
  {\dot \xi}_3 &= \tfrac{4B}{V_t}(ku\sqrt{P_d-\text{sign}(u)\xi_3}-\tfrac{\alpha \xi_3}{1+\gamma \lvert u\rvert}-S\xi_2 ).
  \end{split}
\end{equation}
Let
\begin{equation} \label{MyEq5}
  \begin{split}
  e_1 &=\xi_1-r,\\
  e_2 &=\xi_2-{\dot r},\\
  e_3 &=f(\xi)-{\ddot r},
  \end{split}
\end{equation}
where $f(\xi) = {\dot \xi_2}$, the error dynamics satisfy
\begin{equation} \label{MyEq6}
  \begin{split}
  {\dot e_1} &= e_2 + d,\\
  {\dot e_2} &= e_3,\\
  {\dot e_3} &= (\tfrac{\partial f(\xi)}{\partial \xi}){\dot \xi}-{\dddot r}.
  \end{split}
\end{equation}
\section{Adaptive Control Law Design}
More often than not, EHSS parameters are subject to variations due to temperature rise. Owing to the fact that backstepping controller
 design relies on actual system's parameters, the need arises to design a controller that adapts to these changes. To overcome the problem of variation in parameters, adaptive control schemes are generally employed. In this section, we propose a backstepping-based adaptive technique that is robust to uncertainties in the system's parameters and external disturbance. We assume the parameters of the load (that is $\beta$) and b are unknown nonlinear functions whose parameters will be estimated by the adaptive scheme. The schematic of the proposed backstepping-based adaptive strategy is shown in Fig~\ref{adaptivecontrol} below.
\begin{figure}[t]
  \centering
  \captionsetup{justification=centering,margin=2cm}
  \includegraphics[width=13cm, height=6cm]{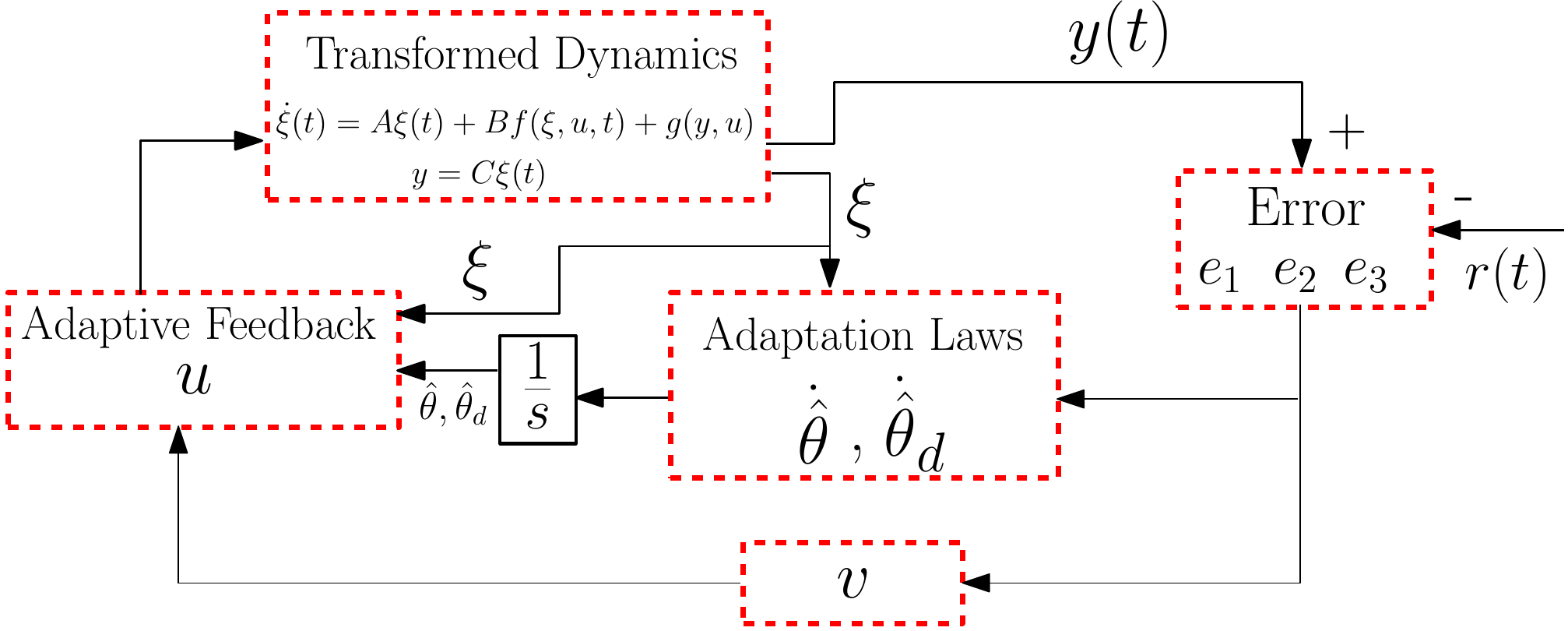}
  \caption{Adaptive Control Scheme}
  \label{adaptivecontrol}
\end{figure}

We also considered a scenario whereby a more complicated vibrator-ground model can result due to the non-ideal contact stiffness that may exist at the boundary interaction between the vibrator's baseplate and ground as depicted in Fig.~\ref{vibrator}. In order to achieve this, we replaced $\beta$ and $b$ as given in \eqref{MyEq85}.
\begin{figure}[t]
  \centering
  \includegraphics[width=8cm]{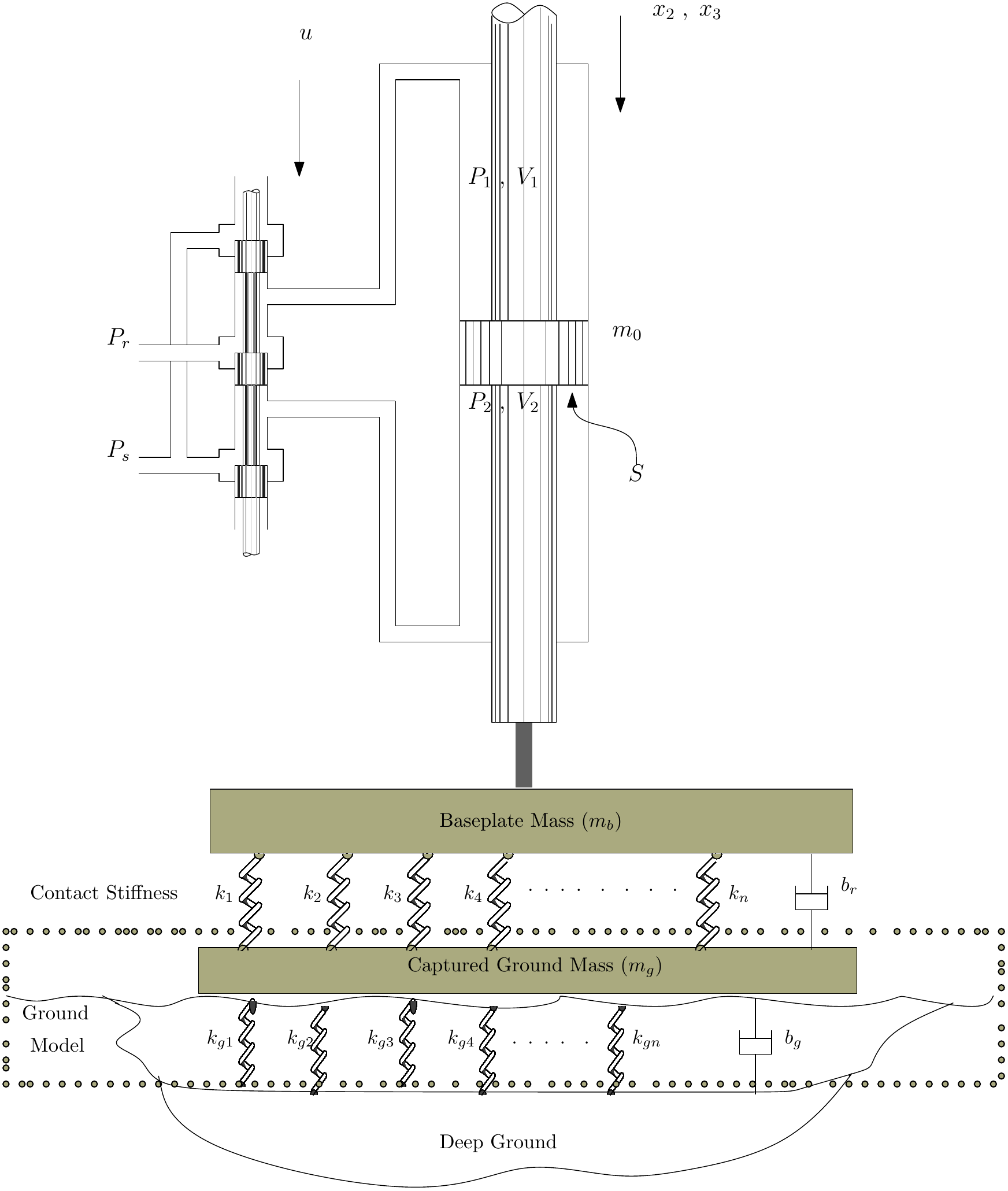}
  \caption{A more detailed vibrator-ground model prototype \cite{wei2010modelling}}
  \label{vibrator}
\end{figure}
\begin{equation} \label{MyEq85}
  \begin{split}
 \beta &=\gamma_3\xi_1^2+\gamma_4\xi_2^2+\gamma_5\xi_3^2 = \theta^T\phi(\xi)\\
 b &= b_0+\Delta f(\xi,b_0)
  \end{split}
\end{equation}
where $\Delta f(\xi,b_0)$ is unknown but bounded nonlinear function that satisfies \eqref{MyEq85a}
\begin{equation} \label{MyEq85a}
  \begin{split}
 \sup_{t \geq0}|\Delta f(\xi,b_0)|\leq F_{max}
  \end{split}
\end{equation}
Now, the error dynamics are given as:
\begin{equation} \label{MyEq87a}
  \begin{split}
  {\dot e_1} &= e_2 + d,\\
  {\dot e_2} &= e_3,\\
  {\dot e_3} &= \tfrac{b_0}{m_t^2}\theta^T\phi(\xi)\xi_1 +(-\tfrac{1}{m_t}\theta^T\phi(\xi)+\tfrac{b_0^2}{m_t^2}-\tfrac{4BS^2}{m_tV_t})\xi_2\\
  &+(-\tfrac{b_0S}{m_t^2}-\tfrac{4BS\alpha}{m_tV_t(1+\gamma \lvert u \rvert)})\xi_3 + \Delta F_1\theta^T\phi(\xi)\xi_1\\
  & +\Delta F_2\xi_2-\Delta F_3\xi_3-\theta_d^T\phi(\xi)d(t)-{\dddot r} + Am(t)u.
  \end{split}
\end{equation}
where, $\theta_d^T\phi(\xi)=\tfrac{\beta}{m_t}$ and $\xi$ is a state vector comprising of $\xi_1$, $\xi_2$ and $\xi_3$. Now, the goal is to design an adaptive feedback such that:
\begin{equation} \label{MyEq74}
  \begin{split}
 \lim_{t\rightarrow\infty} |\xi_1-r(t)| \leq \delta
  \end{split}
\end{equation}
where, $\delta$ is a sufficiently small positive number. The design task is to make $\delta$ as small as possible and at the same time ensuring smooth control law. The following gives the control design schemes.
\begin{equation} \label{MyEq87}
  \begin{split}
  \Delta F_1 &=\tfrac{\Delta f(\xi,b_0)}{m_t^2} \leq F_{max}\\
  \Delta F_2  &= \tfrac{1}{m_t^2}(2b_0\Delta f(\xi,b_0)+\Delta f^2(\xi,b_0)) \leq F^3_{max}\\
  \Delta F_3  &= \tfrac{S}{m_t^2}\Delta f(\xi,b_0) \leq F_{max}
  \end{split}
\end{equation}
\begin{theorem}
Given that
\begin{align} \label{MyEq14}
& \alpha_4= \tfrac{3}{2}+\lambda ,\nonumber \\
& \alpha_5= 1 + \tfrac{1}{\lambda^3}+\tfrac{1}{2 \lambda},\\
& \alpha_6=\tfrac{1}{\lambda^4}\nonumber.
\end{align}
and
\begin{equation} \label{MyEq99}
  \begin{split}
   h(e) &= \alpha_4e_1+\alpha_5e_2+\alpha_6e_3\\
   g(e,\xi) &= e_2+\lambda e_1+\alpha_4e_2+\alpha_5e_3\\
   &+\alpha_6(\tfrac{b_0^2}{m_t^2}-\tfrac{4BS^2}{m_tV_t})\xi_2-\tfrac{b_0S}{m_t^2}\xi_3\\
   A_1 & =\tfrac{\alpha_6b_0}{m_t^2}, \ \ \ A_2=\tfrac{\alpha_6}{m_t} \ \ \  A_3=\tfrac{4BS\alpha\alpha_6}{m_tV_t}\\
   \tilde{\theta} & = \theta -\hat \theta
  \end{split}
\end{equation}
and let the adaptation law be given as
\begin{equation}
  \begin{split}\label{MyEq99aaa}
  {\dot {\hat \theta}} &=\gamma_6(A_1h(e)\phi(\xi)\xi_1-A_2h(e)\phi(\xi)\xi_2\\
                       &+\alpha_6F_{max}|\xi_1|\phi(\xi))\\
  {\dot {\hat \theta}_d} &=-\gamma_7\alpha_6\phi(\xi)d_{max}
  \end{split}
\end{equation}
and let the adaptive feedback be given as
\begin{equation}\label{MyEq2ccc}
  \begin{split}
  u=\bigg(\tfrac{m_tV_t}{4\alpha_3SBk \text{min}(\sqrt{P_d-\xi_3},\sqrt{P_d+\xi_3}}\bigg)v
  \end{split}
\end{equation}
where,
\begin{equation}
  \begin{split}
  v=&-(|g(e,\xi)|+|\alpha_4-\alpha_6{\hat \theta}^T_d\phi(\xi)|d_{max}+\Phi(\xi,\hat \theta)\\
    &+\alpha_6|\xi_3|F_{max}+\alpha_6|\xi_2|F^3_{max})-k_oh(e)
  \end{split}
\end{equation}
Then, system \eqref{MyEq4} under the adaptive feedback control law given in \eqref{MyEq2ccc} is practically stable and the solution of the error dynamic \eqref{MyEq87a} is globally uniformly ultimately bounded with ultimate bound satisfying the following condition
\begin{equation}
  \begin{split}
||e||^2 \leq \frac{d_{max}^2}{\lambda\,\sigma_{min}(\phi \phi^T)}\leq\frac{d_{max}^2}{2 \left(\lambda+\lambda^{5}\right)\,\sigma_{min}(\phi \phi^T)}
\end{split}
\end{equation}
with
\begin{equation} \label{phi2}
  \begin{split}
\phi = \left[ \begin{array}{lll} 1&\lambda&\frac{\alpha_1}{\lambda^2}\\ 0 & 1 & \frac{\alpha_2}{\lambda^2}\\ 0& 0&\frac{\alpha_3}{\lambda^2} \end{array}\right]
\end{split}
\end{equation}
\end{theorem}
\begin{proof}
To prove the boundedness of the error dynamics, we again choose the Lyapunov functions as follows:
\begin{equation} \label{MyEq88}
  \begin{split}
  V_1 &= \tfrac{1}{2}e_1^2,\ \ \ V_2=\tfrac{1}{2\lambda^4}(e_2+\lambda e_1)^2\\
   V_3 &=\tfrac{1}{2}(\alpha_4 e_1 + \alpha_5 e_2 + \alpha_6 e_3)^2\\ V_4 &=\tfrac{1}{2\gamma_6}{\tilde{\theta}}^T{\tilde{\theta}}^T+\tfrac{1}{2\gamma_7}{\tilde{\theta_d}}^2,
  \end{split}
\end{equation}
and again using the Young's inequality with  $\lambda >0$ then,
\begin{equation} \label{MyEq103}
  \begin{split}
  {\dot V}= {\dot V_1}+{\dot V_2}+{\dot V_3}+{\dot V_4}
  \end{split}
\end{equation}
Therefore, the derivative of the Lyapunov function is thus given:
\begin{equation} \label{MyEq104}
  \begin{split}
  {\dot V} & \leq  -\tfrac{\lambda}{2} e_1^2+ (\tfrac{1}{2\lambda}+\tfrac{1}{2\lambda^5})d^2-(e_2+ \lambda e_1)^2\\
  &+ h(e)[g(e,\xi)+A_1\tilde{\theta}^T\phi(\xi)\xi_1+A_1{\hat\theta}^T\phi(\xi)\xi_1\\
  &-A_2\tilde{\theta}^T\phi(\xi)\xi_2-A_2{\hat\theta}^T\phi(\xi)\xi_2-\tfrac{A_3}{1+\gamma \lvert u \rvert})\xi_3 \\
  &+ \alpha_4d+  \alpha_6\Delta F_1\tilde{\theta}^T\phi(\xi)\xi_1+\alpha_6\Delta F_1{\hat\theta}^T\phi(\xi)\xi_1\\
  &+\alpha_6\Delta F_2\xi_2-\alpha_6\Delta F_3\xi_3-\alpha_6\tilde{\theta}_d^T\phi(\xi)d(t)\\
  &-\alpha_6{\hat\theta}_d^T\phi(\xi)d(t)-\alpha_6{\dddot r} + \alpha_6Am(t)u]\\
   &-\tfrac{1}{\gamma_6}{\tilde{\theta}}^T{\dot {\hat \theta}}-\tfrac{1}{\gamma_7}\tilde{\theta}_d{\dot {\hat \theta}_d}
  \end{split}
\end{equation}
In order to annihilate the parametric error, the update laws are chosen as in \eqref{MyEq99aaa}
Given that $F_{max}=\sup_{t \geq 0}|\Delta F_1|$ and $d_{max}= \sup_{t \geq 0}|d(t)|$, and to ensure a uniformly ultimately bounded tracking error, adaptive feedback $``u"$ is chosen as given in \eqref{MyEq2ccc} and $v$ is given as follows:
\begin{equation} \label{MyEq110}
  \begin{split}
  v&=-(|g(e,\xi)|+|\alpha_4-\alpha_6{\hat \theta}^T_d\phi(\xi)|d_{max}+\Phi(\xi,\hat \theta)\\
  &+\alpha_6|\xi_3|F_{max}+\alpha_6|\xi_2|F^3_{max})-k_oh(e)
  \end{split}
\end{equation}
Ultimately,
\begin{equation} \label{MyEq111}
  \begin{split}
{\dot V}  &\leq -\tfrac{\lambda}{2} e_1^2+ (\tfrac{1}{2\lambda}+\tfrac{1}{2\lambda^5})d^2-(e_2+ \lambda e_1)^2\\
&-(\alpha_1 e_1 + \alpha_2 e_2 + \alpha_3 e_3)^2
  \end{split}
\end{equation}
Therefore,for ${\dot V} \leq 0$, it is sufficient to verify that $-V  +\tfrac{1}{\lambda}d_{max}^2 \leq 0$. Which means that ${\dot V} \geq 0$ if $e$ is such that $V \leq \tfrac{1}{\lambda}d_{max}^2$. This will lead to increasing $e$ until $V \geq \tfrac{1}{\lambda}d_{max}^2$.\\
Let $z=\phi^T e$, with
\begin{equation} \label{weights1}
  \begin{split}
\phi = \left[ \begin{array}{lll} \frac{1}{\sqrt{2}}&\lambda&\alpha_1\\ 0 & 1 & \alpha_2\\ 0& 0&\alpha_3 \end{array}\right]
\end{split}
\end{equation}
Using (\ref{weights1}), $V$ can be rewritten as $V=z^T\,z=||z||^2$. On the other hand,
\begin{equation} \label{weights2}
  \begin{split}
\sigma_{min}(\phi \phi^T)||e||^2\leq V=||z||^2 \leq \sigma_{max}(\phi \phi^T)||e||^2
\end{split}
\end{equation}
where $\sigma_{min}(\phi \phi^T)$ and $\sigma_{max}(\phi \phi^T)$ represent the max and min singular values of $\phi \phi^T$ respectively.
Since the Lyapunov function satisfies \eqref{MyEq111} for all $t \in \Re_{\geq0}$, this implies that the whole time during which the adaptation take place is finite. During the finite time, the variables $\hat \theta$ and ${\hat \theta}_d$ cannot escape to infinity since the adaptation laws in \eqref{MyEq99aaa} are well defined. For any bounded conditions $e(0)$, ${\hat \theta}(0)$ and ${\hat \theta}_d(0)$ and by making use of \eqref{MyEq111}, we infer that $e$, ${\hat \theta}$ and ${\hat \theta}_d$ are bounded for all $t \in \Re_{\geq0}$. The proof ends here.
\end{proof}
\begin{remark}
The parameter $\lambda$ is a design parameter introduced in Young's inequality. Such parameter should be selected as big as possible to decrease the ultimate bound in the tracking error dynamic. Nevertheless, a trade-off must be made, since the larger the value of $\lambda$ the more oscillatory is the transient and higher is the control input. However, $\epsilon$ has to be chosen sufficiently small and its choice is neither dependent on system's parameters nor the bound of the disturbance.
\end{remark}
The proposed controller has parameters that should be carefully selected for optimal performance of the closed loop system. In this work, the ABC technique is employed for optimal setting of the controller parameters $\lambda$ and $\gamma_1$ based on a preassigned objective function minimization. The proper selection of these parameters will be reflected on minimizing the objective function. The proposed objective function is defined by the error and control signal.
\begin{equation}
  Obj = \sum\limits_{t=0}^{t_{sim}}\big(\Gamma_1 e_1^2(t)+\Gamma_2 u^2(t)\big)
\end{equation}
where the tracking error $e_1(t)=\xi_1-r(t)$ and $u(t)$ is the control signal, $\Gamma_1$ and $\Gamma_2$ are weighting parameters. The controller's parameters are selected within the bounds:
\begin{equation*} \label{MyEq88888}
  \begin{aligned}
   \lambda^{min}  \leq & \lambda \leq \lambda^{max} \\
   \gamma_1^{min} \leq & \gamma_1 \leq \gamma_1^{max} \\
  \end{aligned}
\end{equation*}
\section{ABC Algorithm}
Artificial Bees Colony is a meta-heuristic approach that gained its inspiration from the work proposed in \cite{Karaboga_2005}. The algorithm is motivated from the way life is structured in the colony of natural bees. The bees in the colony are usually divided into three groups: employed, onlooker and scout bees. The primary assignment of the employed bees is to randomly search for food and the best solution of the food is identified as the optimal solution. The employed bees dance in a systematic manner for the purpose of relaying information about the food source and the amount of nectar to other bees in the colony. Onlooker bees differentiate between the good and the bad food sources based on the dance length, dance type and speed of shaking of the employed bees. Onlooker bees also use these information to establish the quality of food. The scout bees are chosen from the onlooker bees prior a new search for food. The onlooker and scout bees may decide to switch roles with the employed bees depending on the quality of food \cite{Karaboga_2007}.\\
\indent
As detailed in \cite{Karaboga_2007}, the employed and onlooker bees are responsible for searching the solution space for the optimal parameters while the scout bees control the search process. The summary of the ABC algorithm is shown in the flowchart of Fig.~\ref{abc1} and the solution of the optimization problem is the position of the food source while the amount of nectar with respect to the quality is termed as the objective function of the optimization procedure. The position of the food source in the search space can be described as follows:
\begin{equation} \label{MyEq15a}
  \begin{split}
  x^{new}_{ij}=x^{old}_{ij}+u(x_{ij}^{old}-x_{kj})
  \end{split}
\end{equation}

\begin{figure}[htb!]
  \centering
  \includegraphics[width=8cm,height= 13cm]{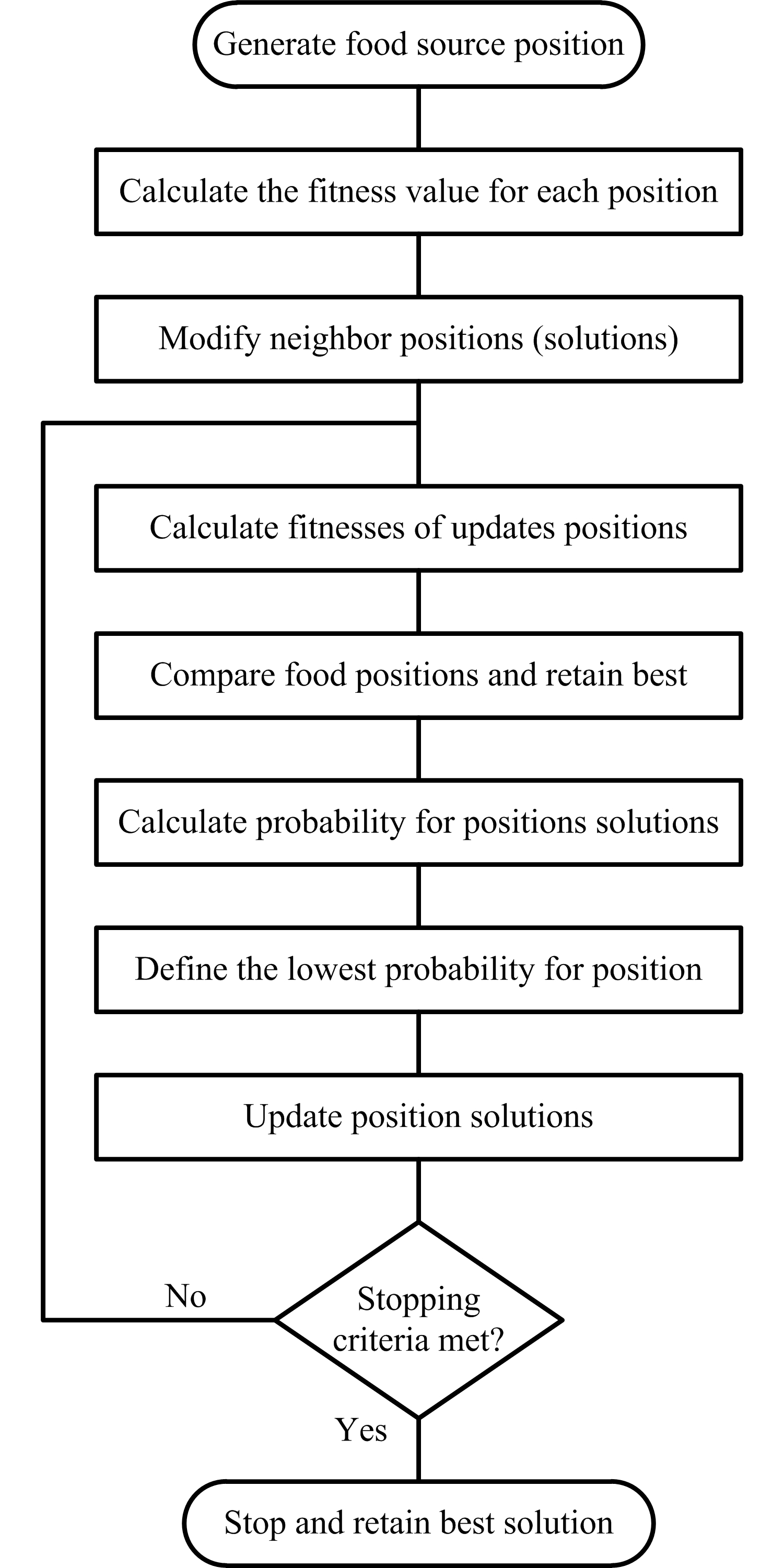}
  \caption{Artificial Bee Colony Algorithm}\label{abc1}
\end{figure}
The probability of onlooker bees for choosing a food source:
\begin{equation} \label{MyEq86}
  \begin{split}
  P_i=\tfrac{fitness_1}{\sum^{E_b}_{i=1}fitness_i}
  \end{split}
\end{equation}
where, $x$ is a candidate solution, $P_i$ is the probability of onlooker solution,  $i=1,2,\cdots,E_b$ is the half of the colony size, $j=1,2,\cdots,D$ and and $j$ is the number of positions with D dimension where $D$ refers to number of parameters to be optimized, $fitness_i$  is the fitness function, $k$ is a random number where $k \in (1,2,\cdots,E_b)$, $u$ is random number between 0 and 1.\\*[.3pc]
\section{Results and Discussions}
The problem is divided into two stages. In the first stage, the proposed backstepping controller (BSC) was implemented without ABC algorithm and in the second stage, ABC algorithm is incorporated with the proposed controller for optimal parameter tuning. The tracking capability of both stages is evaluated using constant, sum of sinusoidal, and sinusoidal reference inputs. The objective function of the optimization procedure is formulated as:
\begin{equation} \label{MyEq888}
  \begin{split}
  Obj= \sum^{20}_{t=0.01}\Gamma_1 e_1^2(t)+\Gamma_2 u^2(t)
  \end{split}
\end{equation}
\begin{equation*} \label{MyEq8888}
  \begin{aligned}
   9 \leq & \lambda \leq 16\\
   10^{-7} \leq & \gamma_1 \leq 10^{-10}\\
  \end{aligned}
\end{equation*}
In the optimization algorithm, the number of parameters to be optimized are $\lambda$ and $\gamma_1$, the population size is selected as 50 and the number of generations is 100, and the search space was constrained as given in \eqref{MyEq888}. The weighted values of $\Gamma_1$ and $\Gamma_1$ were chosen to be 1. For experimental purpose, $F_{max}$ is chosen to be $10$. Adaptive backstepping controller with ABC-based optimizer is deployed on the system model in \eqref{MyEq4}. It must be remarked that this is a minimization task. The parameters of the system's model are given in table {\ref{table:Table2}}\\
\begin{table}[h]
\setlength{\tabcolsep}{1pt}
\caption{Numerical values for simulations} 
\centering 
\small
\begin{tabular}{c c c} 
\hline\hline 
Parameters & Value & Units \\ [0.5ex]
\hline\hline 
$B$ & $2.2e9$ & Pa \\[1ex] 
$P_r$ & $1e5$ & Pa\\[1ex] 
$V_t$ & $1e{-3}$ & $m^3$ \\[1ex]
$S$ & $1.5e{-3}$ & $m^2$ \\[1ex]
$\gamma$ & 8571 & $s^{-1}$ \\[1ex]
$b$ & 590 & $kg~s^{-1}$ \\[1ex]
$\Delta$$k_l$ & 2500 & $N m^{-1}$ \\[1ex]
 $k_l$ & 12500 & $N m^{-1}$ \\[1ex]
$P_s$ & $300e5$ & Pa \\[1ex]
 $m_t$ & 70 & kg\\[1ex]
$k$ & $5.12e{-5}$ & $m^3 s^{-1} A^{-1} Pa^{1/2}$ \\[1ex]
$\alpha$ & $4.1816e-12$ & $m^3 s^{-1} Pa^{-1}$ \\[1ex]
\hline\hline 
\end{tabular}
\label{table:Table2} 
\end{table}

\begin{table*}[!t]
 \setlength{\tabcolsep}{3pt}
 \caption{Minimum objective function with optimal parameters.} 
\centering 
\small
\begin{tabular}{|c| c| c| c|c| c| c| c|c|c|} 
\hline\hline 
Experiment No &   1   &     2    &   3    &   4  &   5   &   6   &   7  &   8  \\ [0.5ex]
\hline 
$Objective$ &   2.1368  &  2.1368   &  2.1368  &  2.1368 & 2.1369  &  2.1368   &  2.1368  &  2.1369   \\[1ex] 
\hline
$\gamma_1$ &  $10^{-10}$  &    $10^{-10}$   &   $10^{-10}$    &  $10^{-10}$  &   $10^{-10}$   &   $10^{-10}$   &   $10^{-10}$  &  $10^{-10}$ 	 \\[1ex]
\hline 	   	   	   	   				
$\lambda$&  13.5585   &   13.5580    &  13.5583  &  13.5585  &   13.5580   &   13.5585  &  13.5581 &  13.5589  \\[1ex]	  	  	  	  																					 
\hline\hline 
\end{tabular} 						
\label{table:result2} 
\end{table*}
The simulation is implemented for 20 seconds and the values of error and control signal are captured every 0.01 second. In the first set of experiments, the reference input is a step function with $r(t)=0.2$. The simulation was carried out 8 times with different initial populations in order to ensure the robustness of the proposed solution. As shown in Table~\ref{table:result2}, the objective functions were close in all experiments and Fig.~\ref{ABCOUT2} shows the output performance of the proposed BSC with optimized parameters $\gamma_1 \leq 10^{-10}$ and $\lambda \leq 13.5585$. Although steady state error is not equal to zero due to the disturbance ($d(t) = 0.1$), the error and control signal are bounded and the output performance is very close to the reference. The performance of the proposed controller was also compared with sliding mode controller (SMC) and it can be seen that even though the SMC achieved smaller tracking error, there is output chattering even at steady state. However in the case of the proposed controller tuned with ABC, the control signal is smooth and the transient oscillations have been eliminated. ABC was able to find the best compromise solution by simultaneously minimizing the tracking error and output oscillations in both transient and steady state. Also, as shown in Fig.~\ref{ABCOUT3} the control signal is minimized and smooth for ABC-based BSC compared to both SMC and BSC without ABC tuning.
\begin{figure}[!h]
  \centering
  \includegraphics[width=9cm]{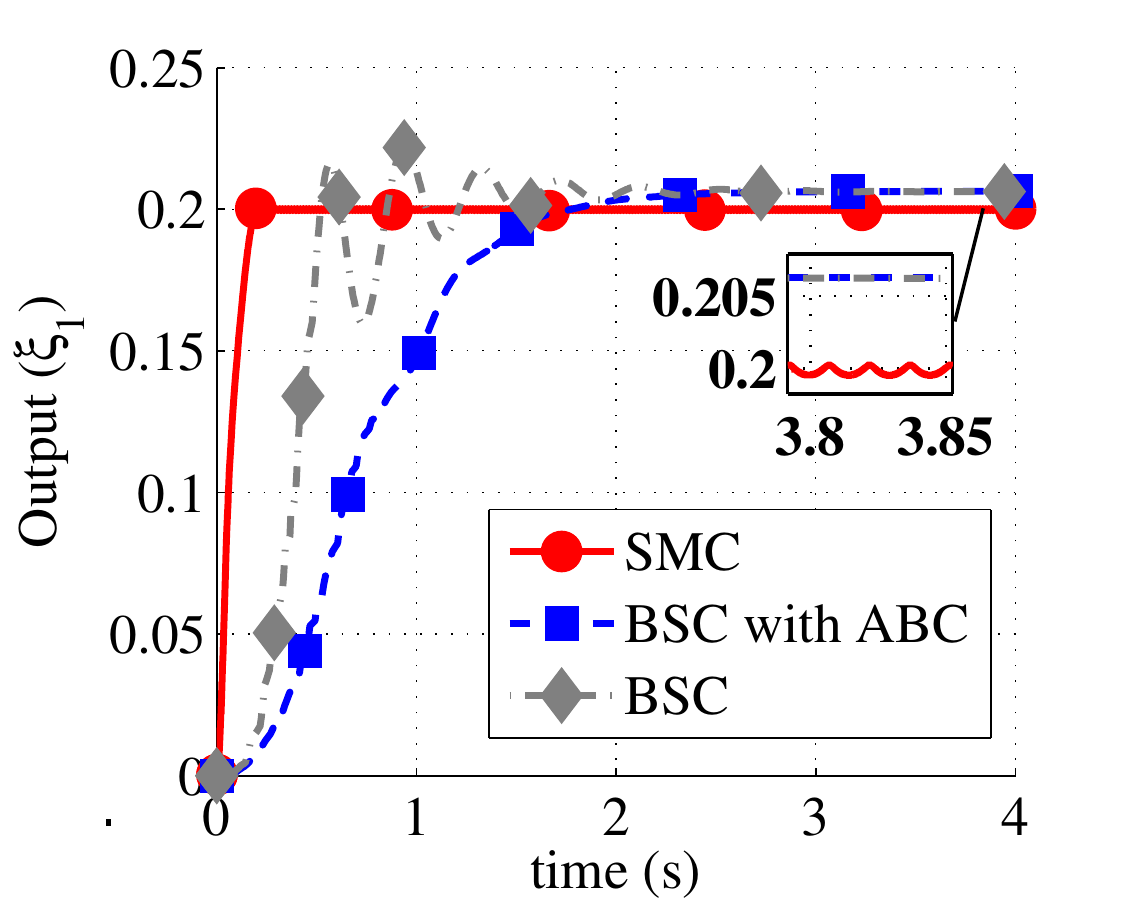}
  \caption{Output performance of adaptive backstepping after ABC optimization with $r = 0.2$m}\label{ABCOUT2}
\end{figure}

\begin{figure}[!h]
  \centering
  \includegraphics[width=9cm]{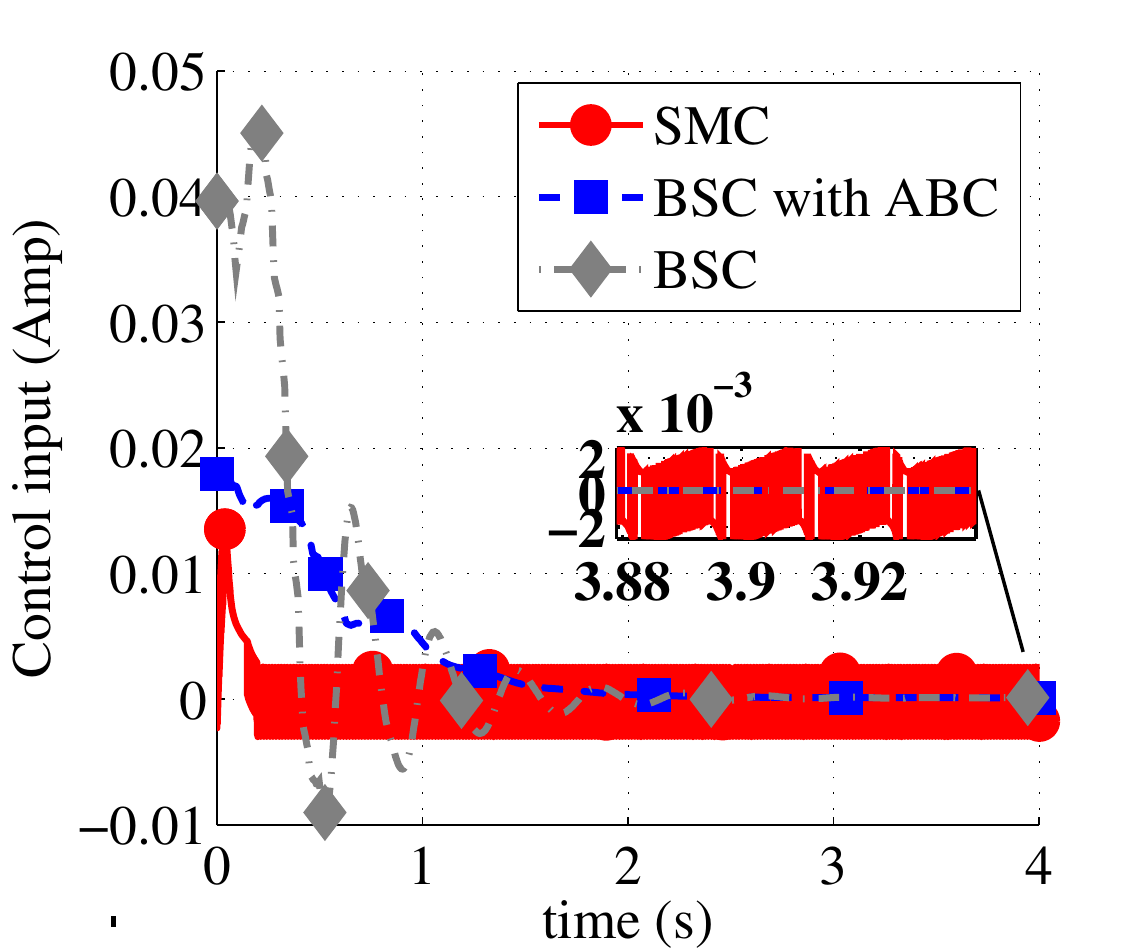}
  \caption{Control input of adaptive backstepping after ABC optimization with $r = 0.2$m}\label{ABCOUT3}
\end{figure}

\begin{figure}[!h]
  \centering
  \includegraphics[width=9cm]{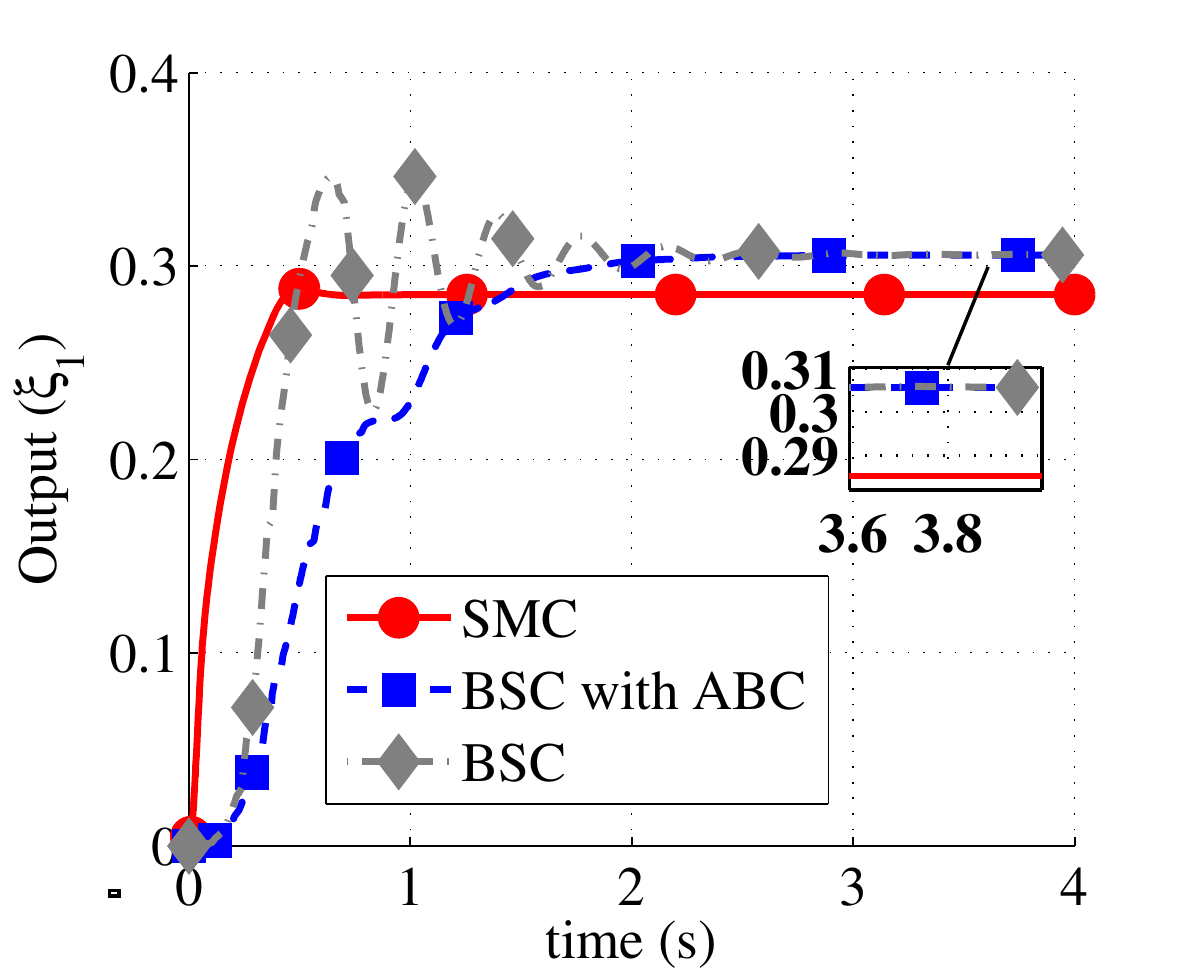}
  \caption{Output performance of adaptive backstepping after ABC optimization with $r = 0.3$m}\label{ABCOUT4}
\end{figure}

\begin{figure}[!h]
  \centering
  \includegraphics[width=9cm]{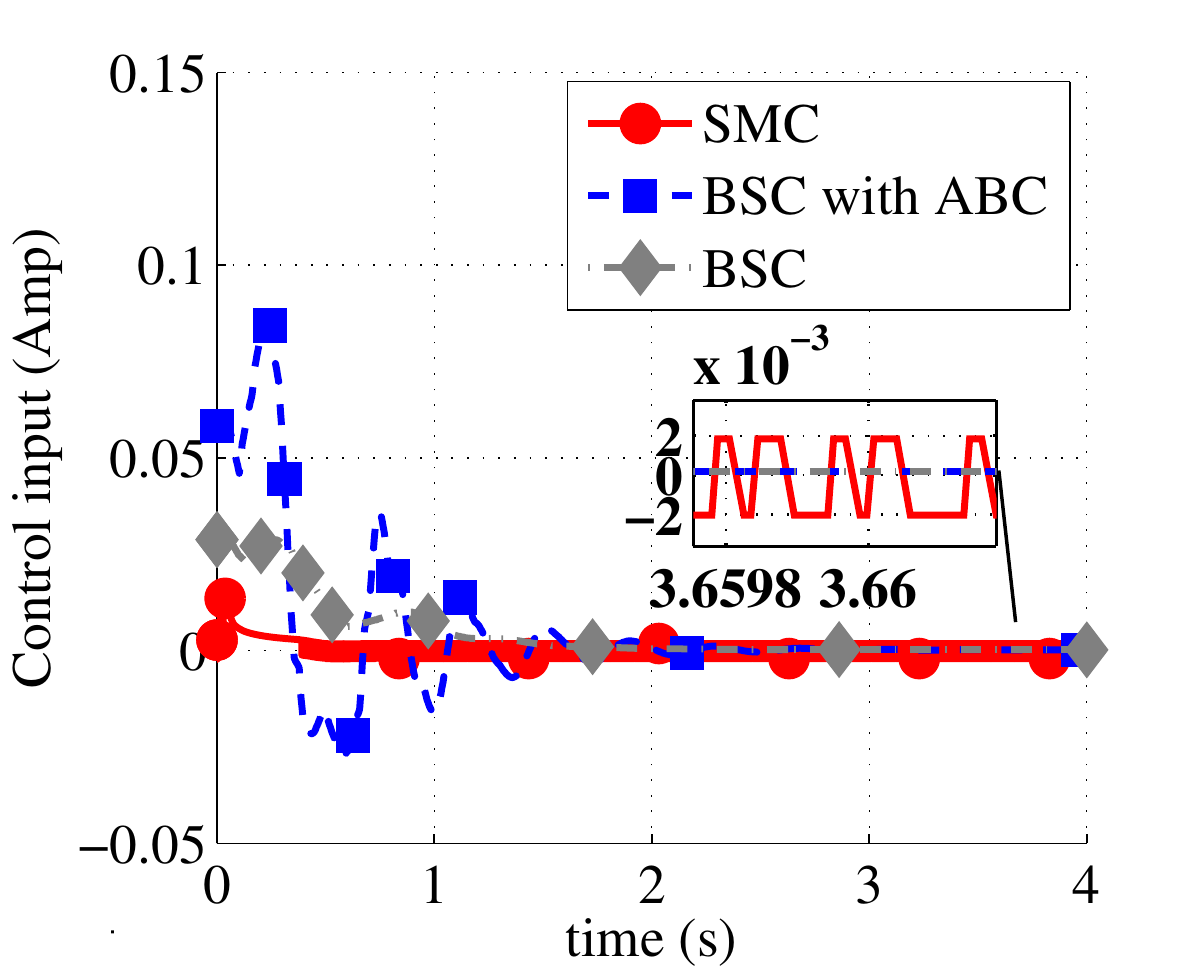}
  \caption{Control input of adaptive backstepping after ABC optimization with $r = 0.3$m}\label{ABCOUT5}
\end{figure}
\begin{figure}[!h]
  \centering
  \includegraphics[width=9cm]{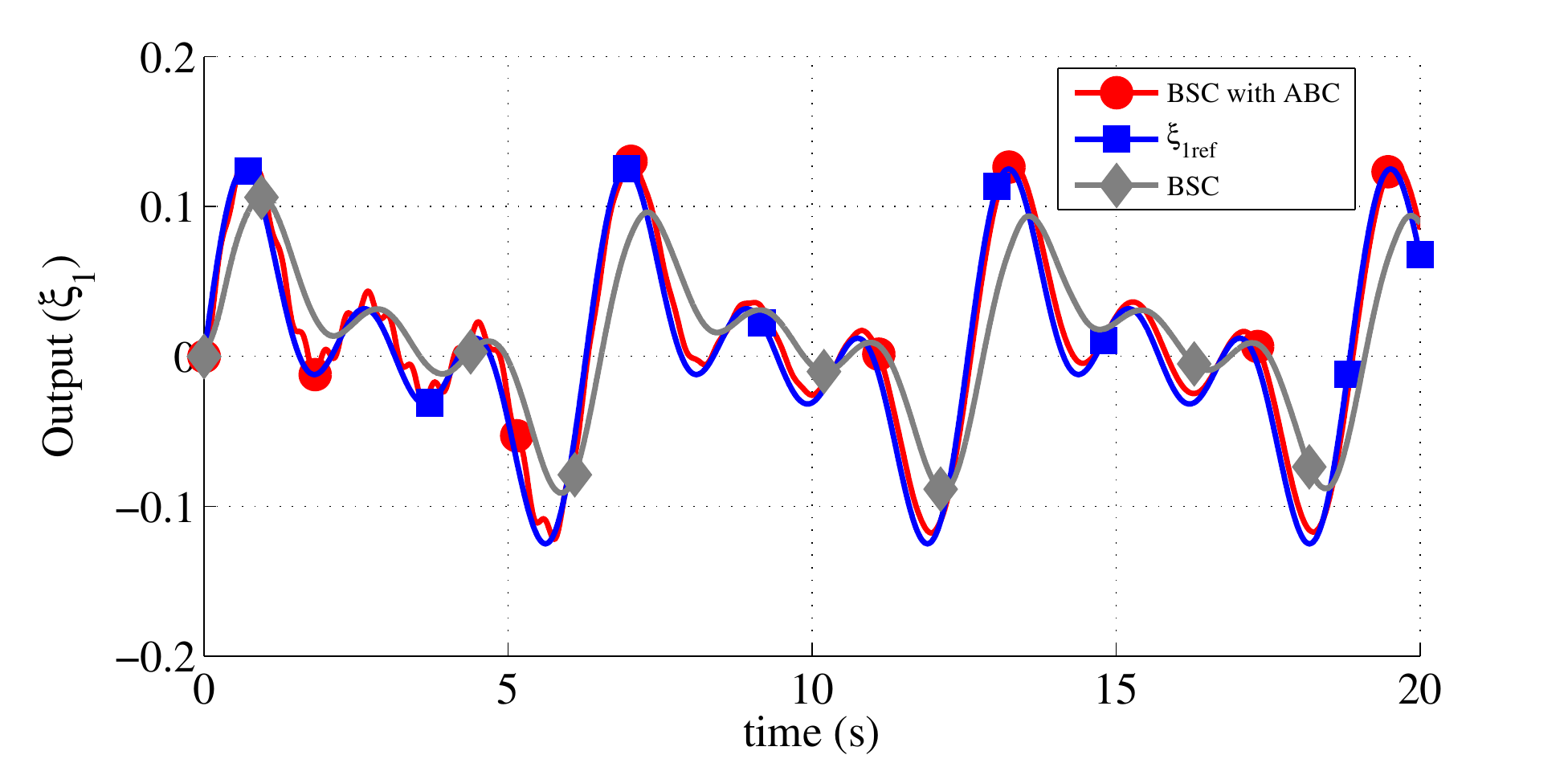}
  \caption{Output performance of adaptive backstepping after ABC optimization with $r=0.05(sin(t)+sin(2t)+sin(3t))$m}\label{ABCOUT6}
\end{figure}
\begin{figure}[!h]
  \centering
  \includegraphics[width=9cm]{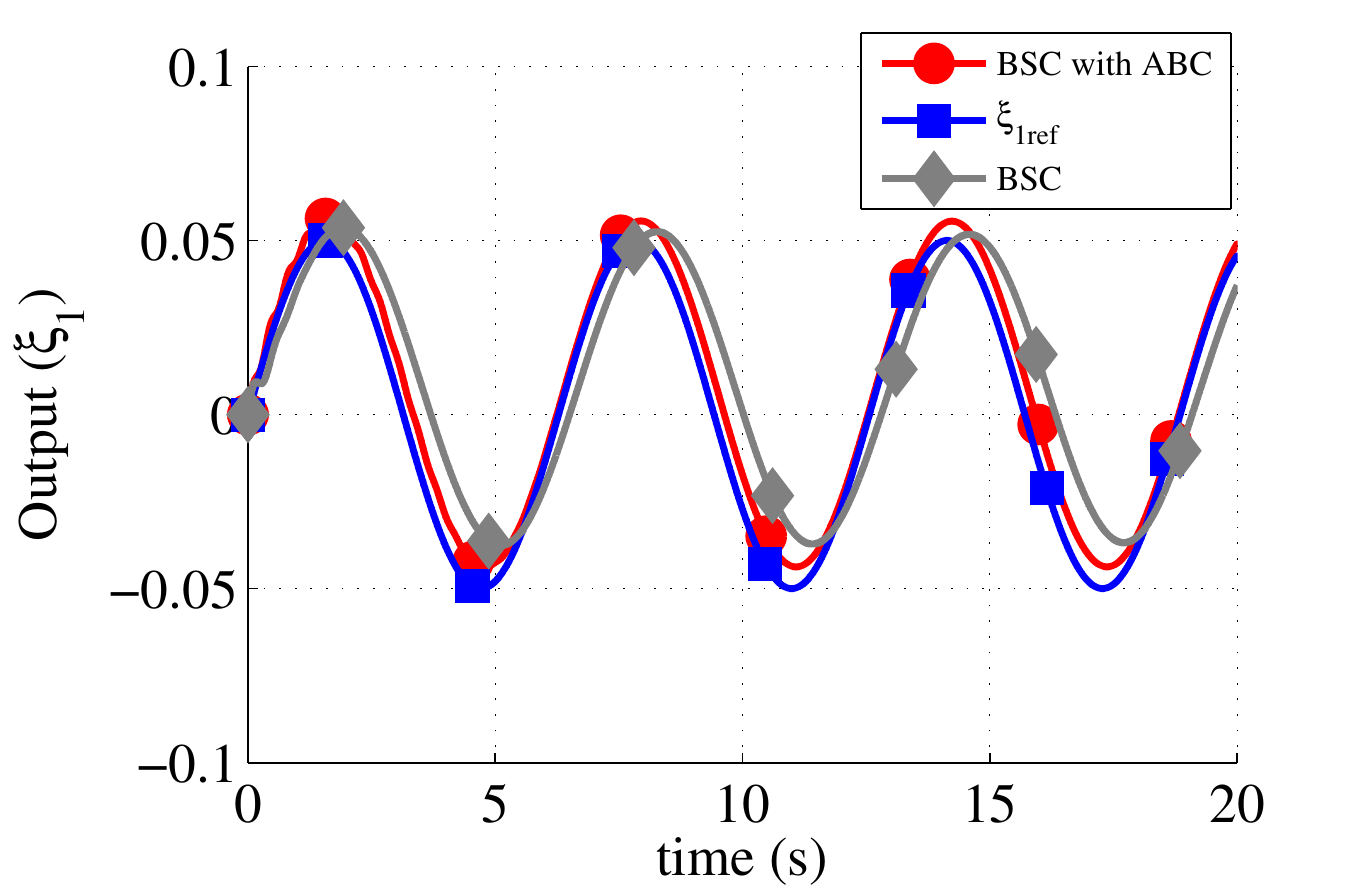}
  \caption{Output performance of adaptive backstepping after ABC optimization with $r=0.05sin(t)$m}\label{ABCOUT7}
\end{figure}
It must be noted that the proposed adaptive design does not involve the differentiation of $m(t)$, which indicates the proposed scheme can handle the effects of various types of slowly time-varying $m(t)$ and $d(t)$. The problem we address is a robust adaptive control issue for EHSS. Using a constant trajectory as seen in Fig.~\ref{ABCOUT2}, the adaptive BSC input accomplishes a bounded error tracking even in the presence of input nonlinearity, parameter uncertainties and unknown but bounded disturbance. Fig.~\ref{ABCOUT4} shows the output performance when the reference was changed to $r(t)=0.3$ and it can be seen that BSC tuned with ABC was able to achieve better tracking accuracy, smooth response, and less transient oscillations than both SMC and BSC. It can again be observed in Fig.~\ref{ABCOUT5} that the SMC control signal is non-smooth and those of BSC are smooth.\\
\indent
Fig.~\ref{ABCOUT6} demonstrates the excellent tracking of the proposed ABC-based BSC for sum of sinusoids reference input compared to BSC without ABC. It is worth mentioning that for sinusoidal references, the SMC went out of stability and for this reason, we could not compare its response with BSC for experiments involving sinusoidal references. Also in Fig.~\ref{ABCOUT7}, excellent tracking accuracy was noticed for ABC-based BSC compared to BSC when the reference was changed to sinusoid. In a nutshell, the proposed solution is robust against parameter uncertainties and disturbance if the ultimate bound satisfies the condition $||e||^2 \leq \frac{d_{max}^2}{2 \left(\lambda+\lambda^{5}\right)\,\sigma_{min}(\phi \phi^T)}$.  It must also be remarked that the results presented in this paper are for the case when non-ideal contact stiffness exists at the boundary interaction between the vibrator's baseplate and ground.  The $\beta$ and $b$ which correspond to load and the frictional parameters have been replaced with nonlinear functions to incorporate more realistic vibrator-ground model. As shown, the tracking error still converges to the neighborhood of the origin at the steady states. This shows the robustness of the backstepping-based adaptive controller.
\section{CONCLUSION}
In this work, a robust backstepping based adaptive controller is proposed for EHSS with uncertain and partially known parameters.
ABC algorithm is incorporated in the closed loop system to optimize the proposed controller's parameter and the adaptation gain while
minimizing the tracking error, control signal and its smoothness. Several experiments have been used to validate the tracking capability of the proposed controller and the robustness of the proposed approach. It is concluded that the proposed control approach ensures uniform ultimate boundedness of the error and control signal.
\bibliographystyle{IEEEtran}
\bibliography{BACKSTEPP}
\end{document}